\documentclass[12pt]{article}
\usepackage[margin=1in]{geometry}
\usepackage{url}

\usepackage[T1]{fontenc}

\usepackage{amsmath}
\usepackage{amssymb}
\usepackage{amsthm}
\usepackage{mathrsfs}

\usepackage{pgf,tikz}
\usetikzlibrary{arrows, automata, shapes, positioning}
\usetikzlibrary{decorations}
\usetikzlibrary{snakes}
\usepackage{verbatim}
\usepackage{multirow}
\usepackage{xcolor}
\usepackage{thmtools}
\usepackage{rotating}

\theoremstyle{plain}
\newtheorem{theorem}{Theorem}
\newtheorem{corollary}[theorem]{Corollary}
\newtheorem{lemma}[theorem]{Lemma}

\theoremstyle{definition}

\newtheorem{example}[theorem]{Example}

\theoremstyle{remark}

\title{Cobham's Theorem and Automaticity}
\author{
Lucas Mol and Narad Rampersad\thanks{The author is supported by an NSERC Discovery Grant.}\\
Department of Mathematics and Statistics\\
University of Winnipeg\\
\url{{l.mol, n.rampersad}@uwinnipeg.ca}\\
\and
Jeffrey Shallit\thanks{The author is supported by an NSERC Discovery Grant.}\\
School of Computer Science\\
University of Waterloo\\
\url{shallit@uwaterloo.ca}\\
\and
Manon Stipulanti\thanks{The author is supported by FRIA Grant 1.E030.16.}\\
Department of Mathematics\\
University of Li\`ege\\
\url{m.stipulanti@uliege.be}}
\date{\today}

\begin{document}
\maketitle

\renewcommand\thmcontinues[1]{Continued}

\begin{abstract}
We make certain bounds in Krebs' proof of Cobham's theorem explicit
and obtain corresponding upper bounds on the length of a common prefix
of an aperiodic $a$-automatic sequence and an aperiodic $b$-automatic
sequence, where $a$ and $b$ are multiplicatively independent.
We also show that an automatic sequence cannot have arbitrarily
large factors in common with a Sturmian sequence.
\end{abstract}

\section{Introduction}
This paper is concerned with the following question:  Given a
$b$-automatic sequence $f$ and a sequence $g$ from some other family
of sequences $\mathscr{G}$, how similar can $f$ and $g$ be?
By ``similar'' we could mean several things:
\begin{enumerate}
\item $f$ and $g$ are identical;
\item $f$ and $g$ have a long common prefix;
\item $f$ and $g$ have a factor of length $n$ in common for infinitely many $n$;
\item $f$ and $g$ have the same set of factors of length $n$ for all sufficiently large $n$;
\item $f$ and $g$ agree on a set of positions of density $1$.
\end{enumerate}
When $\mathscr{G}$ is the family of $a$-automatic sequences, where
$a$ and $b$ are \emph{multiplicatively independent} ($a$ and $b$
are not powers of the same integer), then we have some answers.
Notably, Cobham's theorem \cite{C} states that $f$ and $g$ can be
identical only if $f$ and $g$ are ultimately periodic.  Recently,
Krebs \cite{Kre18} has given a very short and elegant proof of
Cobham's theorem.  Much of what we do in the first part of this paper
is based on this proof of Cobham's theorem.  We also note that Byszewski and
Konieczny \cite{BK17} generalized Cobham's theorem by showing that if
$f$ and $g$ coincide on a set of positions of density $1$, then they are periodic
on a set of positions of density $1$.

One of the main results of this paper concerns the ``long common
prefix'' measure of similarity.  In particular we give explicit bounds
(in terms of the number of states of the automata generating the
sequences) on how long $f$ and $g$ can agree before they are forced to
agree forever.  As an example of a result of this type, consider the
following generalization of the Fine--Wilf theorem
\cite[Theorem~2.3.5]{Sha09}: If $f \in w\{w,x\}^\omega$ and $g \in
x\{w,x\}^\omega$ ($w$ and $x$ are finite words) agree on a prefix of
length $|w|+|x|-\gcd(|w|,|x|)$, then $f=g$.  (Here the notation
$\{w,x\}^\omega$ denotes the set of infinite words of the form
$U_1U_2U_2\cdots$, where each $U_i \in \{w,x\}$.)  In our setting,
where $f$ is an $a$-automatic sequence and $g$ is a $b$-automatic
sequence, we obtain our bounds on the length of the common prefix by
following the proof of Krebs and making explicit several of the bounds
that appear in this proof.  Our result answers a question posed by
Zamboni (personal communication), who asked how long a sequence
generated by a $b$-uniform morphism and one generated by an
$a$-uniform morphism can agree before the two sequences are forced to
be equal.

This problem of bounding the length of the common prefix of $f$ and
$g$ is related to the concept of \emph{$b$-automaticity} of infinite
sequences \cite{Sha96}, which measures the minimum number of states of
a base-$b$ automaton that computes the length-$n$ prefix of the sequence.  In
particular, we are able to get a lower bound on the
$b$-automaticity of an $a$-automatic sequence.

Regarding the property of having ``arbitrarily large factors in
common'', it is not difficult to see that even distinct aperiodic $a$-automatic 
and $b$-automatic sequences can have arbitrarily large
factors in common.  For example, the characteristic sequences of
powers of $2$ and $3$ respectively are $2$-automatic and $3$-automatic
respectively, and clearly have arbitrarily large runs of $0$'s in
common.  The problem in this case is to show that in general such large factors
necessarily have some simple structure; however, we do not address
this question in this paper.

If we now change the family $\mathscr{G}$ of sequences from
$a$-automatic to Sturmian, then it is somewhat easier to answer these
kinds of questions.  \emph{Sturmian sequences} are those given by the first
differences of sequences of the form
$$(\lfloor n \alpha  + \beta\rfloor)_{n \geq 1},$$ where
$0 \leq \alpha, \beta < 1$ and $\alpha$ is irrational \cite{BS}.  The
number $\alpha$ is called the \emph{slope} of the Sturmian sequence
and the number $\beta$ is the called the \emph{intercept}.
It is well-known that a Sturmian sequence cannot be $b$-automatic.
This follows from the fact that the limiting frequency of $1$'s in a
Sturmian sequence is $\alpha$, whereas if a letter in a $b$-automatic
sequence has a limiting frequency, that frequency must be rational
\cite[Thm.~6, p.~180]{C}.

The problem of determining the maximum length of a common prefix of a
$b$-automatic sequence and a Sturmian sequence was examined by
Shallit \cite{Sha96}.  Upper bounds on the length of the common prefix
can be deduced from the automaticity results given by Shallit.
In the present paper we answer, in the negative, the question, ``Can a
Sturmian sequence and a $b$-automatic sequence have arbitrarily large
finite factors in common?''

Byszewski and Konieczny \cite{BK17} examine these questions for the
family of \emph{generalized polynomial functions} (these are sequences defined
by expressions involving algebraic operations along with the floor
function).  This family contains the family of Sturmian sequences as a
subset.  In recent work \cite{BK18}, they have extended some of the
results of this paper to this more general class.

We also mention the work of Tapsoba \cite{Tap95}.  Recall that the
complexity of a word $s$ is the function counting the number of
distinct factors of length $n$ in $s$.  It is also well-known that
Sturmian words have the minimum possible complexity $n+1$ achievable
by an aperiodic infinite word.  Tapsoba shows another distinction
between automatic sequences and Sturmian words by giving a formula for
the minimal complexity function of the fixed point of an injective
$k$-uniform binary morphism and comparing this to the complexity
function of Sturmian words.

\section{Common prefix of $a$-automatic and $b$-automatic\\ sequences}\label{sec:cobham}
This section is largely based on the work of Krebs \cite{Kre18} and so
we will mostly stick to the notation used in his paper.  The reader
should read this section in conjunction with Krebs' paper; we
occasionally omit details that can be found there.

\subsection{Definitions and notation}
Let $b \geq 2$ and let $w \in \{0,1,2,\ldots\}^*$.  Write
$w = w_{n-1} w_{n-2} \cdots w_0$, where each
$w_i \in \{0,1,2,\ldots\}$.  We define the number $[w]_b$ by
\[
[w]_b = w_{n-1} b^{n-1} + w_{n-2} b^{n-2} + \cdots + w_1 b + w_0.
\]
Typically, one restricts $w$ to be over the \emph{canonical digit set}
$\{0,1,\ldots,b-1\}$, in which case every natural number $x$ has a unique
representation $w$ such that $x = [w]_b$ and $w$ does not begin with a $0$ (the number $0$ is represented by the empty string).
In this case, we use $\langle x \rangle_b$ to denote this
representation $w$.

However, Krebs' proof requires the use of a larger digit set.
Let $D_b$ denote the digit set $\{0,\ldots,2b\}$.  Over this digit set,
numbers may no longer have unique representations, even with the
restriction that the representation must begin with a non-zero digit.
We use the notation $(x)_{D_b}$ to refer to some particular representation
of $x$ over the digit set $D_b$ that does not begin with the
digit zero, without necessarily specifying which
representation it is.  Note also that if some representation $(x)_{D_b}$
has length $n$, then
\[
x \leq 2b\sum_{i=0}^{n-1}b^i = \frac{2b(b^n-1)}{b-1} \leq 2b^{n+1}.
\]

A \emph{deterministic finite automaton with output} (DFAO) is a
$6$-tuple $(S, D, \delta, s_0, \Delta, F)$, where $S$ is a finite set
of \emph{states}, $D$ is a finite \emph{input alphabet},
$\delta : S \times D \to S$ is the \emph{transition function},
$s_0 \in S$ is the \emph{initial state}, $\Delta$ is a finite
\emph{output alphabet}, and $F : S \to \Delta$ is the \emph{output
  function}.  See \cite{AS03} for more details.

Let $D$ be a set of non-negative digits containing
$\{0,1,\ldots,b-1\}$.  A sequence $(f_x)_{x \in \mathbb{N}}$ is
\emph{$(b,D)$-automatic} if there is a DFAO
$M = (S, D, \delta, s_0, \Delta, F)$ such that
$f_{[w]_b} = F(\delta(s_0, w))$ for all $w \in D^*$.  Note that for
each $x$, the DFAO $M$ must produce the same output for all
$w \in D^*$ satisfying $x = [w]_b$.  The DFAO $M$ is called a
\emph{$(b,D)$-DFAO}.  A sequence is \emph{$b$-automatic} if it is
$(b,\{0,1,\ldots,b-1\})$-automatic, and the automaton $M$ in this case
is called a \emph{$b$-DFAO}.

\subsection{Normalization}\label{sec:normalization}
Krebs begins his proof by showing that a sequence $f$ is automatic with
respect to representations over the canonical base-$b$ digit set if
and only if it is automatic with respect to representations over the
digit set $D_b$.  The reverse direction can be seen by noting that
given a $(b,D_b)$-DFAO generating $f$, one obtains a $b$-DFAO
generating $f$ simply by deleting the transitions on all digits other
than $\{0,1,\ldots,b-1\}$.  The forward direction is proved using two results:
the first is a modification of \cite[Theorem~6.8.6]{AS03} and the second
can be found in \cite[Proposition~7.1.4]{Lot02}.
The first result \cite[Theorem~6.8.6]{AS03} states that if a sequence
$f$ is generated by a $b$-DFAO $M$, then so is the
sequence obtained by first applying a transducer $T$ to the input and
then feeding the output of $T$ to $M$.  As presented in \cite{AS03},
this result requires $T$ to map words over the digits set
$\{0, 1, \ldots, b-1\}$ to words over the same digit set; however, the
proof is easily modified to allow $T$ to map words over any digit set
to words over $\{0,1,\ldots,b-1\}$.  Krebs therefore applies this modified
version of \cite[Theorem~6.8.6]{AS03} where $T$ is the
transducer of \cite[Proposition~7.1.4]{Lot02}, which
converts input over a non-canonical digit set (in our case $D_b$) to
the canonical digit set for a given base $b$ (this is called
\emph{normalization}).  The result of this operation is therefore a
$(b, D_b)$-DFAO computing $f$.  We now discuss the details of this
construction with the aim of obtaining a reasonably small
$(b,D_b)$-DFAO computing $f$.

Let $N$ be the transducer of \cite[Lemma~7.1.1]{Lot02}, which converts
from the digit set $D_b$ to the digit set $\{0,1,\ldots,b-1\}$ and reads its input from least significant
digit to most significant digit.  The
number of states of $N$ is determined by the quantity
\[m = \max \{|e-d| : e \in D_b, d \in \{0,1,\ldots,b-1\}\};\] in
particular, the state set of $N$ is defined to be
$Q = \{s \in \mathbb{N} : s < m/(b-1)\}$.  In our case, we have
$m=2b$, and furthermore, for $b=2$ we have $2b/(b-1) = 4$ and for
$b>2$ we have $2 < 2b/(b-1) \leq 3$.  We therefore set $\gamma = 4$ if
$b=2$ and $\gamma = 3$ if $b > 2$, so that
$Q = \{s \in \mathbb{N} : s < \gamma\}$.

The set of transitions of $N$ is
\[
E = \{ s \stackrel{e|d}{\longrightarrow} s' : s+e = bs'+d\}.
\]
The initial state is $0$ and the output function $\omega$ maps each
state $s \in Q$ to $\langle s \rangle_b$.  Note that $N$ is
\emph{subsequential}, or ``input-deterministic''.
To see this, suppose we have two transitions
\[
s \stackrel{e|d'}{\longrightarrow} s' \text{ and }
s \stackrel{e|d''}{\longrightarrow} s''.
\]
Then $bs'+d' = bs''+d''$, which we can rewrite as $(s'-s'')b =
d''-d'$.  However, we have $|d''-d'|<b$, so $|s'-s''|<1$, which
implies $s'=s''$ and $d'=d''$.

On input $u = e_ne_{n-1} \cdots e_0$ over $D_b$, the transducer $N$
produces output $v = \omega(s)d_nd_{n-1}\cdots d_0$ over
$\{0,1,\ldots,b-1\}$, where $s$ is the state reached by $N$ after
reading $u$, and $[u^R]_b = [v^R]_b$.

\begin{example}[label=exa:cont]\label{ex:base2-TM-1}
Throughout this section, we illustrate the proof with the case $b=2$. 
In this case, the transducer $N$ is the one given in Figure~\ref{fig:trans}.
For instance, on input $u = 4032$ over $D_2$, the transitions of $N$ are 
\[
0 \stackrel{2|0}{\longrightarrow} 1 \stackrel{3|0}{\longrightarrow} 2 \stackrel{0|0}{\longrightarrow} 1 \stackrel{4|1}{\longrightarrow} 2,
\]
so $N$ outputs $v = \langle 2 \rangle_2 1000 = 10 1 000$, which is the canonical base-$2$ expansion of $u$.
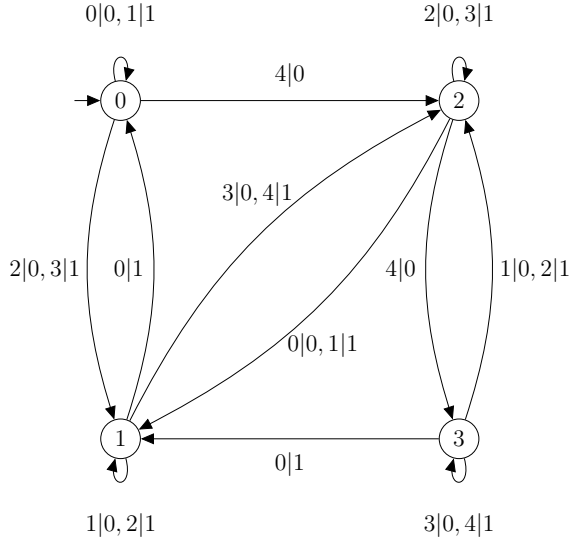
\begin{figure}
\begin{center}
\scalebox{0.75}{
\begin{tikzpicture}
\tikzstyle{every node}=[shape=circle,fill=none,draw=black,minimum size=20pt,inner sep=2pt]
\node[](a) at (0,0) {$0$};
\node[](b) at (0,-6) {$1$};
\node[](c) at (6,0) {$2$};
\node[](d) at (6,-6) {$3$};

\tikzstyle{every node}=[shape=circle,fill=none,draw=none,minimum size=10pt,inner sep=2pt]
\node(a0) at (-1,0) {};

\tikzstyle{every path}=[color =black, line width = 0.5 pt, > = triangle 45]
\tikzstyle{every node}=[shape=circle,minimum size=5pt,inner sep=2pt]
\draw [->] (a0) to [] node [] {}  (a);

\draw [->] (a) to [loop above] node [] {$0|0,1|1$}  (a);
\draw [->] (a) to [bend right=18] node [left] {$2|0, 3|1$}  (b);
\draw [->] (a) to node [above] {$4|0$}  (c);

\draw [->] (b) to [loop below] node [below] {$1|0,2|1$}  (b);
\draw [->] (b) to [bend right=18] node [left] {$0|1$}  (a);
\draw [->] (b) to [bend left=18] node [above] {$3|0, 4|1$}  (c);

\draw [->] (c) to [loop above] node [above] {$2|0,3|1$}  (c);
\draw [->] (c) to [bend right=18] node [left] {$4|0$}  (d);
\draw [->] (c) to [bend left=18] node [below] {$0|0, 1|1$}  (b);

\draw [->] (d) to [loop below] node [] {$3|0,4|1$}  (d);
\draw [->] (d) to [] node [below] {$0|1$}  (b);
\draw [->] (d) to [bend right=18] node [right] {$1|0, 2|1$}  (c);

;
\end{tikzpicture}}
\end{center}
\caption{The transducer $N$ in base $2$ converting $D_2$ into $\{0,1\}$.}\label{fig:trans}
\end{figure}
\end{example}

Let $M = (S, \{0,1,\ldots,b-1\}, \delta, I, \Delta, F)$ be a
$b$-DFAO generating a $b$-automatic sequence $f$.  Recall that our
convention is that a $b$-DFAO reads its input from most significant
digit to least significant digit.

\begin{example}[continues=exa:cont]\label{ex:base2-TM-5}
We now consider the Thue--Morse sequence $t = 01101001\cdots$ which is the fixed point of the morphism $\tau : 0 \mapsto 01, 1 \mapsto 10$. 
It is well known that the Thue--Morse sequence $t$ is $2$-automatic and can be generated by the $2$-DFAO $M=(S, \{0,1\}, \delta, I, \Delta, F)$ with $S=\{0,1\}=F$ and $I=0$ drawn in Figure~\ref{fig:Thue-Morse}.

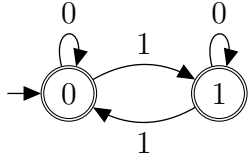
\begin{figure}
\begin{center}
\scalebox{1.0}{
\begin{tikzpicture}
\tikzstyle{every node}=[shape=circle,fill=none,draw=black,minimum size=20pt,inner sep=2pt]
\node[accepting](a) at (0,0) {$0$};
\node[accepting](b) at (2,0) {$1$};

\tikzstyle{every node}=[shape=circle,fill=none,draw=none,minimum size=10pt,inner sep=2pt]
\node(a0) at (-1,0) {};

\tikzstyle{every path}=[color =black, line width = 0.5 pt, > = triangle 45]
\tikzstyle{every node}=[shape=circle,minimum size=5pt,inner sep=2pt]
\draw [->] (a0) to [] node [] {}  (a);

\draw [->] (a) to [loop above] node [] {$0$}  (a);
\draw [->] (a) to [bend left] node [above] {$1$}  (b);

\draw [->] (b) to [loop above] node [] {$0$}  (b);
\draw [->] (b) to [bend left] node [below] {$1$}  (a);
;
\end{tikzpicture}}
\caption{The $2$-DFAO $M$ generating the Thue--Morse sequence.}
\label{fig:Thue-Morse}
\end{center}
\end{figure}

\end{example}

Let
$M' = (S', D_b, \delta', I', \Delta, F'),$ be the $(b,D_b)$-DFAO
defined as follows (again, it reads its input from most significant
digit to least significant digit).
We define
\begin{align*}
S' &= \left\lbrace \left\lbrace (s_0,0),(s_1,1),\ldots,(s_{\gamma-1},\gamma-1)\right\rbrace : s_0,s_1,\ldots,s_{\gamma-1}\in S \right\rbrace, \mbox{ and}\\
I' &= \left\{\,\left(\,\delta(I, \langle q\rangle_b), q\right): 0\leq q<\gamma \,\right\}.
\end{align*}
Clearly we have $I'\in S'$.  For any $t\in S'$ and $e\in D_b$, we define
\begin{align*}
\delta'(t, e) &= \bigcup_{(s,q)
  \in t}\left\{\,\left(\,\delta(s, d), q'\right) : {\textstyle q'
  \stackrel{e|d}{\longrightarrow} q} \mbox{ in } N \,\right\}.
\end{align*}
Finally, for $t \in S'$, define $F'(t) = F(s)$, where $(s,0)\in t$ (by the definition of $S'$, there is a unique such $s\in S$).

We first show that $\delta'$ is well-defined. 
Let $t\in S'$ and $e\in D_b$, and we will show that $\delta'(t,e)\in S'$.  We need to show that for every state $p$ of $N$ (i.e., every $p\in Q$) the set $\delta'(t,e)$ contains a unique element of the form $(s,p)$, where $s\in S$.  Let $p\in Q$ be a state of $N$.  Since $N$ is input-deterministic, there is exactly one outgoing transition from $p$ in $N$ with input symbol $e$, say $p\stackrel{e|d}{\longrightarrow}q$ in $N$.  Since $(s,q)\in t$ for exactly one $s\in S$ (by definition of $S'$), we conclude that $(\delta(s,d),p)\in \delta'(t,e)$, and it is the unique element in $\delta'(t,e)$ with second coordinate $p$.

Now we show that $M'$ computes the same automatic sequence as $M$.  For any $u=u_m\cdots u_0 \in D_b^*$ that doesn't begin with $0$, there exists exactly one
$v=v_n\cdots v_0 \in \{0,1,\ldots,b-1\}^*$ that doesn't begin with $0$ such that
$[u]_b = [v]_b$.  Namely, $v=\langle [u]_b\rangle_b$.  Note that $m\leq n\leq m+2$.  We need to show that if $(s,0) \in \delta'(I', u)$, then $\delta(I,v)=s$.  Suppose that $(s,0)\in \delta'(I',u)$.  Then in $N$, we have
\[
0\stackrel{u_0|v_0}{\longrightarrow} q_0 \stackrel{u_1|v_1}{\longrightarrow} q_1 \stackrel{u_2|v_2}{\longrightarrow}\cdots \stackrel{u_m|v_m}{\longrightarrow} q_m,
\]
and $\langle q_m\rangle_b=v_n\cdots v_{m+1}$.  Therefore, we have $(\delta(I,v_n\cdots v_{m+1}),q_m)\in I'$, and retracing the steps of $M'$, we conclude that
\[
\delta(I,v)=s.
\]

Informally, $M'$ works through the transducer $N$ in the reverse direction, while computing the transitions of $M$ on the output.  Since we are working through the transducer backwards, there are $\gamma$ possible places to start, each corresponding to a different backwards path through the transducer.  Further, if we start working backwards from state $q$ in the transducer, then the output function of the transducer will be $\langle q\rangle_b$.  The output function of the transducer is read first by  $M'$, which explains the definition of $I'$.  Only when we reach the end of the input string do we know which backwards path through the transducer was correct (the one that started at state $0$), so $M'$ computes the transitions of $M$ for all $\gamma$ paths along the way.

We have therefore shown how, given a $b$-DFAO $M$ for $f$, to produce
a $(b, D_b)$-DFAO $M'$ that also generates $f$.  Furthermore, the
$(b,D_b)$-DFAO $M'$ has at most
\begin{equation}\label{conversion_blowup}
|S|^{\gamma} \leq |S|^{4}
\end{equation}
states.

\begin{example}[continues=exa:cont]\label{ex:base2-TM-6}
In Figure~\ref{fig:Thue-Morse-prime}, we give the $(2, D_2)$-DFAO $M'$ (omitting all
unreachable states) that computes the Thue--Morse sequence.  We also give its transition table in Table~\ref{table:Thue-Morse-prime}. 
To that aim, recall that $\gamma=4$.
From Figure~\ref{fig:Thue-Morse}, we also get 
\begin{align*}
  I' &= \left\{\,\left(\,\delta(I, \langle q\rangle_b), q\right): 0\leq q<\gamma \,\right\} \\
  &= \{ (\delta(I,\epsilon),0), (\delta(I,1),1), (\delta(I,10),2), (\delta(I,11),3) \} \\
&=\{ (0,0), (1,1), (1,2), (0,3) \}.
\end{align*}

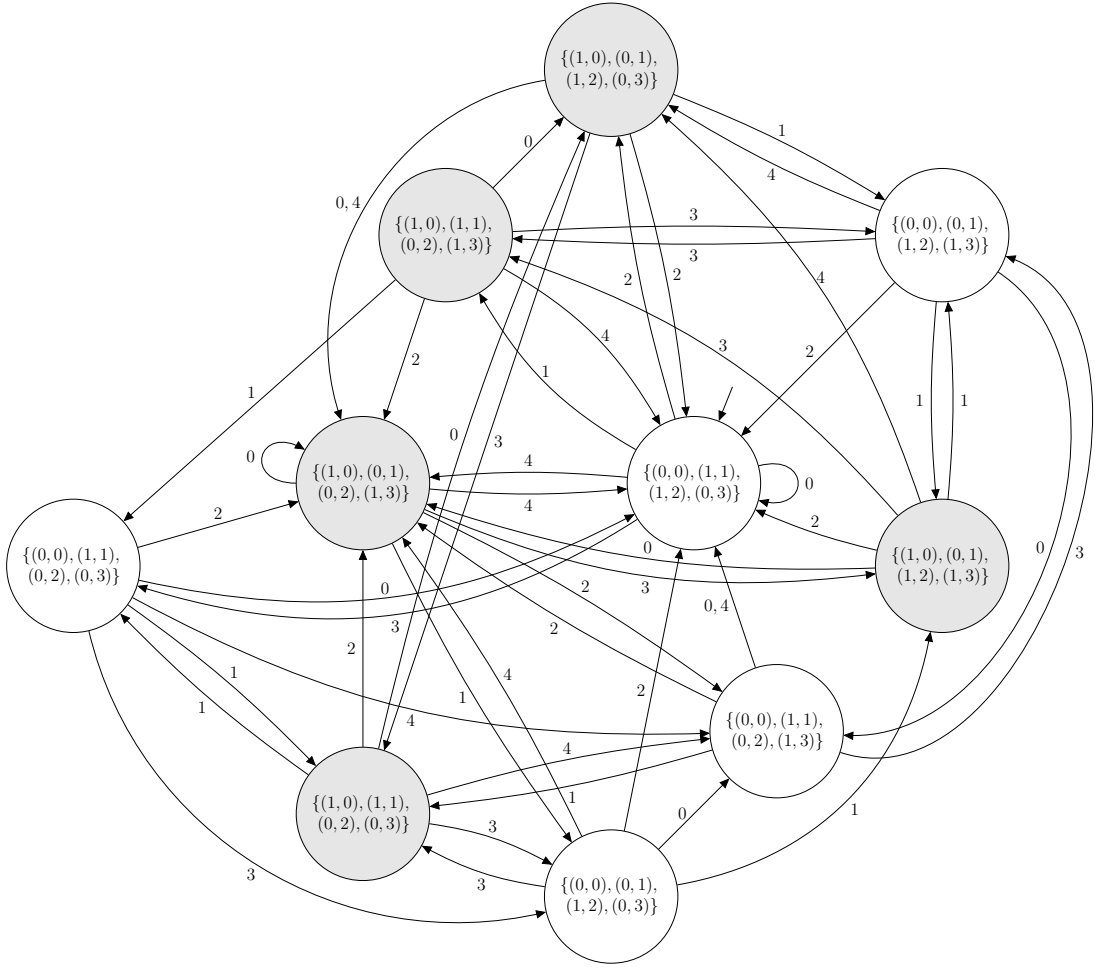
\begin{figure}
\begin{center}
\scalebox{0.55}{
\begin{tikzpicture}
\tikzstyle{every node}=[shape=circle,fill=none,draw=black,minimum size=20pt,inner sep=2pt]

\node[fill=black!10](pink) at (8,-4) {\begin{tabular}{c}
$\left\{ (1,0), (1,1), \right.$  \\ 
$\left. (0,2), (1,3)\right\}$ \\
\end{tabular} };

\node[fill=black!10](lightblue) at (12,0) {\begin{tabular}{c}
$\left\{ (1,0), (0,1), \right.$  \\ 
$\left. (1,2), (0,3)\right\}$ \\
\end{tabular} };

\node[](darkgreen1) at (-1,-12) {\begin{tabular}{c}
$\left\{ (0,0), (1,1), \right.$  \\ 
$\left. (0,2), (0,3)\right\}$ \\
\end{tabular} };

\node[fill=black!10](purple) at (6,-10) {\begin{tabular}{c}
$\left\{ (1,0), (0,1), \right.$  \\ 
$\left. (0,2), (1,3)\right\}$ \\
\end{tabular} };

\node[](red) at (14,-10) {\begin{tabular}{c}
$\left\{  (0,0), (1,1), \right.$  \\ 
$\left. (1,2), (0,3)\right\}$ \\
\end{tabular} };

\node[](brown) at (20,-4) {\begin{tabular}{c}
$\left\{ (0,0), (0,1),  \right.$  \\ 
$\left. (1,2), (1,3)\right\}$ \\
\end{tabular} };

\node[fill=black!10](darkgreen2) at (6,-18) {\begin{tabular}{c}
$\left\{  (1,0), (1,1),  \right.$  \\ 
$\left. (0,2), (0,3)\right\}$ \\
\end{tabular} };

\node[](darkblue) at (16,-16) {\begin{tabular}{c}
$\left\{ (0,0), (1,1),  \right.$  \\ 
$\left. (0,2), (1,3)\right\}$ \\
\end{tabular} };

\node[fill=black!10](black) at (20,-12) {\begin{tabular}{c}
$\left\{ (1,0), (0,1),  \right.$  \\ 
$\left. (1,2), (1,3)\right\}$ \\
\end{tabular} };

\node[](yellow) at (12,-20) {\begin{tabular}{c}
$\left\{  (0,0), (0,1), \right.$  \\ 
$\left. (1,2), (0,3)\right\}$ \\
\end{tabular} };

\tikzstyle{every node}=[shape=circle,fill=none,draw=none,minimum size=10pt,inner sep=2pt]
\node(a0) at (15,-7.5) {};

\tikzstyle{every path}=[color =black, line width = 0.5 pt, > =
triangle 45]
\tikzstyle{every node}=[shape=circle,minimum size=5pt,inner sep=2pt]
\draw [->] (a0) to [] node [] {}  (red);

\draw [->] (pink) to [] node [above] {$0$}  (lightblue);
\draw [->] (pink) to [] node [above left] {$1$}  (darkgreen1);
\draw [->] (pink) to [] node [right] {$2$}  (purple);
\draw [->] (pink) to [bend left=3] node [above] {$3$}  (brown);
\draw [->] (pink) to [bend left=15] node [right] {$4$}  (red);

\draw [->] (lightblue) to [bend right=50] node [left] {$0,4$}  (purple);
\draw [->] (lightblue) to [bend left=5] node [above] {$1$}  (brown);
\draw [->] (lightblue) to [bend left=5] node [right] {$2$}  (red);
\draw [->] (lightblue) to [] node [right] {$3$}  (darkgreen2);

\draw [->] (darkgreen1) to [bend right=20] node [above left] {$0$}  (red);
\draw [->] (darkgreen1) to [bend left=5] node [above right] {$1$}  (darkgreen2);
\draw [->] (darkgreen1) to [] node [above] {$2$}  (purple);
\draw [->] (darkgreen1) to [bend right=45] node [below left] {$3$}  (yellow);
\draw [->] (darkgreen1) to [bend right=15] node [below] {$4$}  (darkblue);

\draw [->] (purple) to [loop left, out=540, in=510, looseness=4] node [left] {$0$}  (purple);
\draw [->] (purple) to [bend right=5] node [left] {$1$}  (yellow);
\draw [->] (purple) to [bend left=5] node [above right] {$2$}  (darkblue);
\draw [->] (purple) to [bend right=15] node [below] {$3$}  (black);
\draw [->] (purple) to [bend right=5] node [below] {$4$}  (red);

\draw [->] (red) to [loop right, looseness=4] node [right] {$0$}  (red);
\draw [->] (red) to [bend left=15] node [above] {$1$}  (pink);
\draw [->] (red) to [bend left=5] node [left] {$2$}  (lightblue);
\draw [->] (red) to [bend left=25] node [below] {$3$}  (darkgreen1);
\draw [->] (red) to [bend right=5] node [above] {$4$}  (purple);

\draw [->] (brown) to [bend left=75] node [left] {$0$} (darkblue);
\draw [->] (brown) to [bend right=5] node [left] {$1$}  (black);
\draw [->] (brown) to [] node [above left] {$2$}  (red);
\draw [->] (brown) to [bend left=3] node [below] {$3$}  (pink);
\draw [->] (brown) to [bend left=5] node [below] {$4$}  (lightblue);

\draw [->] (darkgreen2) to [bend left=5] node [left] {$0$}  (lightblue);
\draw [->] (darkgreen2) to [bend left=5] node [below left] {$1$}  (darkgreen1);
\draw [->] (darkgreen2) to [] node [left] {$2$}  (purple);
\draw [->] (darkgreen2) to [bend left=10] node [above] {$3$}  (yellow);
\draw [->] (darkgreen2) to [bend left=5] node [above] {$4$}  (darkblue);

\draw [->] (darkblue) to [] node [left] {$0,4$}  (red);
\draw [->] (darkblue) to [bend left=5] node [below] {$1$}  (darkgreen2);
\draw [->] (darkblue) to [bend left=5] node [below left] {$2$}  (purple);
\draw [->] (darkblue) to [bend right=90] node [right] {$3$}  (brown);

\draw [->] (black) to [bend left=10] node [above] {$0$}  (purple);
\draw [->] (black) to [bend right=5] node [right] {$1$}  (brown);
\draw [->] (black) to [bend left=5] node [above] {$2$}  (red);
\draw [->] (black) to [bend right=15] node [above] {$3$}  (pink);
\draw [->] (black) to [bend right=15] node [above] {$4$}  (lightblue);

\draw [->] (yellow) to [] node [left] {$0$}  (darkblue);
\draw [->] (yellow) to [bend right=35] node [below right] {$1$}  (black);
\draw [->] (yellow) to [] node [left] {$2$}  (red);
\draw [->] (yellow) to [bend left=10] node [below] {$3$}  (darkgreen2);
\draw [->] (yellow) to [bend right=5] node [above right] {$4$}  (purple);
;
\end{tikzpicture}
}
\caption{The $(2, D_2)$-DFAO $M'$ computing the Thue--Morse sequence
  (``white'' states output $0$; ``grey'' states output $1$).}
\label{fig:Thue-Morse-prime}
\end{center}
\end{figure}
\begin{sidewaystable}
\centering
\begin{tabular}{|c|c|c|c|c|c|c|c|c|c|c|}
\hline
$\delta'(t, e)$ & \multicolumn{3}{c|}{$e\in\{0,1,2\}$}\\
\hline
$t \in S'$ & $0$ & $1$ & $2$\\
\hline 
$\{(0,0), (1,1), (1, 2), (0,3)\}$ & $\{(0,0), (1,1), (1, 2), (0,3)\}$ & $\{(1,0), (1,1), (0, 2), (1,3)\}$ & $\{(1,0), (0,1), (1, 2), (0,3)\}$  \\
$\{(1,0), (1,1), (0, 2), (1,3)\}$  & $\{(1,0),(0,1),(1,2),(0,3)\}$ & $\{(0,0),(1,1),(0,2),(0,3)\}$& $\{(1,0),(0,1),(0,2),(1,3)\}$ \\
$\{(1,0), (0,1), (1, 2), (0,3)\}$ & $\{(1,0),(0,1),(0,2),(1,3)\}$ & $\{(0,0),(0,1),(1,2),(1,3)\}$ & $\{(0,0),(1,1),(1,2),(0,3)\}$ \\
$\{(0,0), (1,1), (0, 2), (0,3)\}$ & $\{(0,0),(1,1),(1,2),(0,3)\}$ & $\{(1,0),(1,1),(0,2),(0,3)\}$& $\{(1,0),(0,1),(0,2),(1,3)\}$ \\
$\{(1,0), (0,1), (0, 2), (1,3)\}$ & $\{(1,0),(0,1),(0,2),(1,3)\}$& $\{(0,0),(0,1),(1,2),(0,3)\}$& $\{(0,0),(1,1),(0,2),(1,3)\}$ \\
$\{(0,0), (0,1), (1, 2), (1,3)\}$ & $\{(0,0), (1,1), (0,2), (1,3)\}$ &  $\{(1,0), (0,1), (1,2), (1,3)\}$ & $\{(0,0), (1,1), (1,2), (0,3)\}$ \\
$\{(1,0), (0,1), (1, 2), (1,3)\}$ & $\{(1,0), (0,1), (0,2), (1,3)\}$ &  $\{(0,0), (0,1), (1,2), (1,3)\}$ & $\{(0,0), (1,1), (1,2), (0,3)\}$ \\
$\{(1,0), (1,1), (0, 2), (0,3)\}$ & $\{(1,0), (0,1), (1,2), (0,3)\}$ &  $\{(0,0), (1,1), (0,2), (0,3)\}$ & $\{(1,0), (0,1), (0,2), (1,3)\}$ \\
$\{(0,0), (0,1), (1, 2), (0,3)\}$ & $\{(0,0), (1,1), (0,2), (1,3)\}$ &  $\{(1,0), (0,1), (1,2), (1,3)\}$ & $\{(0,0), (1,1), (1,2), (0,3)\}$ \\
$\{(0,0), (1,1), (0, 2), (1,3)\}$ & $\{(0,0), (1,1), (1,2), (0,3)\}$ &  $\{(1,0), (1,1), (0,2), (0,3)\}$ &  $\{(1,0), (0,1), (0,2), (1,3)\}$ \\
\hline

\end{tabular}

\bigskip

\begin{tabular}{|c|c|c|}
\hline
$\delta'(t, e)$ & \multicolumn{2}{c|}{$e\in\{3,4\}$}\\
\hline
$t \in S'$ & $3$ & $4$\\
\hline 
$\{(0,0), (1,1), (1, 2), (0,3)\}$ & $\{(0,0), (1,1), (0, 2), (0,3)\}$ & $\{(1,0), (0,1), (0, 2), (1,3)\}$ \\
$\{(1,0), (1,1), (0, 2), (1,3)\}$ & $\{(0,0),(0,1),(1,2),(1,3)\}$ & $\{(0,0),(1,1),(1,2),(0,3)\}$ \\
$\{(1,0), (0,1), (1, 2), (0,3)\}$ & $\{(1,0),(1,1),(0,2),(0,3)\}$ & $\{(1,0),(0,1),(0,2),(1,3)\}$  \\
$\{(0,0), (1,1), (0, 2), (0,3)\}$ & $\{(0,0),(0,1),(1,2),(0,3)\}$ & $\{(0,0),(1,1),(0,2),(1,3)\}$\\
$\{(1,0), (0,1), (0, 2), (1,3)\}$ & $\{(1,0),(0,1),(1,2),(1,3)\}$& $\{(0,0),(1,1),(1,2),(0,3)\}$  \\
$\{(0,0), (0,1), (1, 2), (1,3)\}$ & $\{(1,0), (1,1), (0,2), (1,3)\}$ &  $\{(1,0), (0,1), (1,2), (0,3)\}$  \\
$\{(1,0), (0,1), (1, 2), (1,3)\}$ & $\{(1,0), (1,1), (0,2), (1,3)\}$ &  $\{(1,0), (0,1), (1,2), (0,3)\}$  \\
$\{(1,0), (1,1), (0, 2), (0,3)\}$ & $\{(0,0), (0,1), (1,2), (0,3)\}$ &  $\{(0,0), (1,1), (0,2), (1,3)\}$  \\
$\{(0,0), (0,1), (1, 2), (0,3)\}$ & $\{(1,0), (1,1), (0,2), (0,3)\}$ &  $\{(1,0), (0,1), (0,2), (1,3)\}$  \\
$\{(0,0), (1,1), (0, 2), (1,3)\}$ & $\{(0,0), (0,1), (1,2), (1,3)\}$ &  $\{(0,0), (1,1), (1,2), (0,3)\}$ \\
\hline

\end{tabular}

\caption{The transition function $\delta'$ of $M'$ as a function of $t \in S'$ and $e \in \{0,1,2,3,4\}$.}\label{table:Thue-Morse-prime}
\end{sidewaystable}
We also compute $M'$ on two different words $u \in D_2^*$. 
Take $u = 4032 \in D_2^*$ whose canonical base-$2$ expansion is $v = 101000$. 
The transitions are 
\begin{align*}
I'=\{(0,0),(1,1),\boldsymbol{(1,2)},(0,3)\}
&\overset{4}{\longrightarrow} \{ (1,0), \boldsymbol{(0,1)}, (0,2), (1,3) \} \\
&\overset{0}{\longrightarrow} \{ (1,0), (0,1), \boldsymbol{(0,2)}, (1,3)\} \\
&\overset{3}{\longrightarrow}  \{ (1,0), \boldsymbol{(0,1)}, (1,2), (1,3) \} \\
&\overset{2}{\longrightarrow} \{ \boldsymbol{(0,0)}, (1,1), (1,2), (0,3) \}.
\end{align*}
By definition of $F'$, we have $
F'(\{ (0,0), (1,1), (1,2), (0,3) \}) = F(0) = 0$.
Thus the automaton $M'$ outputs $0$ after reading $u$, just as the automaton $M$ does when reading $v$.  The second coordinates of the ordered pairs in bold are the states of the ``correct path'' through the transducer $N$, in reverse:
\[
0 \stackrel{2|0}{\longrightarrow} 1 \stackrel{3|0}{\longrightarrow} 2 \stackrel{0|0}{\longrightarrow} 1 \stackrel{4|1}{\longrightarrow} 2.
\]
The first coordinate of the bolded pair in $I'$ is $\delta(I,\langle 2\rangle_2)=\delta(I,10)=1$, and the first coordinates of the remaining bolded pairs are determined by starting from state $\delta(I,10)=1$ in $M$ and following the transitions of $M$ given by the output labels of the above path through $N$ (again, working backwards through $N$):
\[
\delta(I,10)=1\stackrel{1}{\longrightarrow}0\stackrel{0}{\longrightarrow}0\stackrel{0}{\longrightarrow}0\stackrel{0}{\longrightarrow}0=\delta(I,v).
\]
This illustrates how, on input $u$, $M'$ computes $F(\delta(I,v))$, which is exactly the output of $M$ on input $v$.

As a second illustration, take $u' = 2014 \in D_2^*$ whose canonical base-$2$ expansion is $v' = 10110$. 
On the input $u'$, the transitions of $M'$ are 
\begin{align*}
I'=\{(0,0),\boldsymbol{(1,1)},(1,2),(0,3)\}
&\overset{2}{\longrightarrow} \{\boldsymbol{(1,0)}, (0,1), (1,2), (0,3)\} \\
&\overset{0}{\longrightarrow} \{(1,0), \boldsymbol{(0,1)}, (0,2), (1,3)\}  \\
&\overset{1}{\longrightarrow} \{(0,0), (0,1), \boldsymbol{(1,2)}, (0,3)\}  \\
&\overset{4}{\longrightarrow} \{ \boldsymbol{(1,0)}, (0,1), (0,2), (1,3)\}.
\end{align*}
Similarly, $F'(\{ (1,0), (0,1), (0,2), (1,3) \}) = F(1) = 1$, so the automaton $M'$ outputs $1$ after reading $u'$, agreeing with the output of $M$ on input $v'$.  Again, we have bolded the ordered pairs corresponding to the ``correct path'' through the transducer $N$.
\end{example}

We end this section with some remarks on the construction.  We hope
that the reader is convinced that the construction we have described works
for any digit set containing $\{0,1,\ldots,b-1\}$ and not just the
digit set $D_b$.  Furthermore, Krebs has pointed out (private
communication) that the number of states needed for the construction
can be improved by changing the digit set from $D_b$ to
$\{0,1,\ldots,2b-2\}$.  Recall that our construction results in a DFAO
with $|S|^\gamma$ states.  If $b=2$, then we have $\gamma=4$, while if $b>2$, then we
have $\gamma=3$.  However, if we change the digit set as suggested by
Krebs, we improve this to $|S|^2$ states.  Krebs' proof of Cobham's
Theorem works just as well with this new choice of digit set; however,
a number of bounds and constants in his proof would have to be
modified.  We do not present these modifications here; we just note
that it is possible to do it.

\subsection{Upper bound on longest commmon prefix}\label{sec:long_prefix}
Having dealt with the conversion to the larger digit set required by
Krebs, we now proceed with the Diophantine approximation result used
by Krebs.

\begin{lemma}\label{dirich}
  Let $a,b \ge 2$ be integers and let $\epsilon$ be a positive
  real number.  Define $$\eta := \max\{\lceil \log_a b \rceil, \lceil
  \log_b a \rceil \}.$$ There are non-negative integers $m,n <
  \eta((b-1)/\epsilon+1)$ such that $|a^m - b^n| \leq \epsilon b^n.$
\end{lemma}

\begin{proof}
  First suppose that $a \geq b$.
  Let $(f_x)_{x \in \mathbb{N}}$ be the sequence such that $a^xb^{-f_x} \in
  [1,b)$ for all $x \in \mathbb{N}$.  Then $0 \leq (\log_b a)x - f_x$, so
  $f_x \leq (\log_b a)x$.  Now by the pigeonhole principle there exist $x<y\leq (b-1)/\epsilon+1$ such that
  $\left|a^yb^{-f_y} - a^xb^{-f_x}\right| \leq \epsilon$; i.e.,
  \[\left|a^{y-x} - b^{f_y-f_x}\right| \leq \epsilon b^{f_y}a^{-x} \leq \epsilon b^{f_y-f_x}.\]
  Thus, we have $m = y-x \leq y \leq (b-1)/\epsilon+1$ and
  \[ n = f_y - f_x \leq f_y \leq (\log_b a)y \leq (\log_b
  a)((b-1)/\epsilon+1) \leq \eta ((b-1)/\epsilon+1), \]
  as required.

  Now suppose that $a<b$.  Applying the previous argument with
  $a^{\lceil\log_a b \rceil}$ in place of $a$ (where $\lceil \rho
  \rceil$ denotes the least integer greater than or equal to $\rho$) , we find that
  \[ m = \lceil\log_a b \rceil (y-x) \leq \lceil\log_a b \rceil y
  \leq \lceil\log_a b \rceil ((b-1)/\epsilon+1) \leq \eta
  ((b-1)/\epsilon+1), \]
  and
  \[n = f_y - f_x \leq f_y \leq \lceil\log_a b \rceil (\log_b a) y
  \leq \lceil\log_a b \rceil (\log_b a)((b-1)/\epsilon+1) \leq \eta
  ((b-1)/\epsilon+1),\]
  as required (the final inequality above follows from the fact that
  $\log_b a <1$ in this case).
\end{proof}

As in Lemma~\ref{dirich}, define $\eta := \max\{\lceil \log_a b
\rceil, \lceil \log_b a \rceil \}$ and also define $\theta := \max\{a,
b\}$.  We now define
\begin{align*}
    E(a,b,R,S) &:= \eta\left[6\left(2\theta^{(S+1)(R+1)}+1\right)(\theta-1)+1\right],\\
    A(a,b,R,S) &:= \left(2\theta^{(S+1)(R+1)}+2\right)\theta^{E(a,b,R,S)},
\end{align*}
and note that both these functions are symmetric under exchange of
their first two arguments and also under exchange of their last two
arguments.

\begin{theorem}\label{thm:common_prefix}
  Let $a,b \geq 2$ be multiplicatively independent integers.
  Let $g=(g_x)_{x \in \mathbb{N}}$ be computed by a
  DFAO $M_a = (S_a, D_a, \delta_a, s_{0,a}, \Delta_a, F_a)$ in base
  $a$ and let $f=(f_x)_{x \in \mathbb{N}}$ be computed by a DFAO
  $M_b = (S_b, D_b, \delta_b, s_{0,b}, \Delta_b, F_b)$ in base $b$.
  Suppose that $f$ and $g$ agree on a prefix of length
  $A(a,b,|S_a|,|S_b|)$.  Then $f$ and $g$ are equal and ultimately
  periodic.
\end{theorem}

\begin{proof}
  Let $S_\infty$ be the subset of states of $M_b$ consisting of all
  states $s$ with the property that there are infinitely many numbers
  $x$ such that some representation $(x)_{D_b}$ reaches state $s$ in
  $M_b$.  For each $s \in S_\infty$, we claim that there must exist a
  state $t \in S_a$ and positive integers $x_{st}$ and
  $y_{st}$ such that some base-$b$ representations
  $(x_{st})_{D_b}$ and $(y_{st})_{D_b}$ both lead to state $s$ in $M_b$ and some
  base-$a$ representations $(x_{st})_{D_a}$ and $(y_{st})_{D_a}$ both lead to state
  $t$ in $M_a$.  We show this by giving an explicit upper bound on
  $x_{st}$ and $y_{st}$.

  If a string $W$ has length at least $|S_b|$, then any computation of
  $M_b$ on $W$ repeats a state.  Since for each $s \in S_\infty$ there
  are infinitely many $(x)_{D_b}$ that reach state $s$, there must exist some
  number $x_0$, some representation $(x_0)_{D_b}$, and some factorization
  $(x_0)_{D_b} = uw$ with the following properties:
  \begin{itemize}
  \item $|(x_0)_{D_b}| \leq |S_b|$.
  \item There exists a non-empty $v$ such that $|v| \leq |S_b|$ and
    $uv^iw$ reaches $s$ for all $i \geq 0$.
  \end{itemize}
  For $1 \leq i \leq |S_a|$, let $x_i$ be the integer such that
  $(x_i)_{D_b} = uv^iw$.  Then the numbers $x_i$ are all distinct.  Now
  consider the states reached in $M_a$ by some choice of representations
  $(x_i)_{D_a}$, for $0 \leq i \leq |S_a|$.  There must be two such
  numbers $x_i$ and $x_j$ such that $(x_i)_{D_a}$ and $(x_j)_{D_a}$ reach the
  same state $t$ in $M_a$. We choose these as our $x_{st}$ and
  $y_{st}$.  Finally, we note that for $0 \leq i \leq |S_a|$, we
  have $|(x_i)_{D_b}| \leq |S_b|(|S_a|+1)$, which gives the bound $x_{st}, y_{st} \leq
  2b^{|S_b|(|S_a|+1)+1}$.

  Let $\xi := \max\{x_{st},y_{st} \,|\, s \in S_\infty\}+1 \leq
  2b^{|S_b|(|S_a|+1)+1}+1$.  By Lemma~\ref{dirich}, there exist
  \[ m,n \leq \eta[6\xi(b-1)+1] \leq E(a,b,|S_a|,|S_b|) \] such
  that $\xi|a^m-b^n| \leq \frac16 b^n$.  As defined in \cite{Kre18},
  let $p_{st} := (x_{st} - y_{st})(a^m-b^n)$ (swapping $x_{st}$ and
  $y_{st}$ if necessary, so that $p_{st}>0$), and note that from
  \cite{Kre18} we have $p_{st} \leq \frac16 b^n$.  Let $z$ be any
  integer such that
  $z, z+p_{st} \in \left[\frac13 b^n, \frac53 b^n\right]$.  In
  particular, there exist representations $(z)_{D_b}$ and $(z+p_{st})_{D_b}$
  such that $|(z)_{D_b}|, |(z+p_{st})_{D_b}| \leq n$.  In what follows, we
  specifically use the representations of $z$ and $z+p_{st}$ that
  satisfy this condition on their lengths.  We also note that by the
  calculation in \cite{Kre18}, we have $z-y_{st}(a^m-b^n) \leq 2a^m$,
  so there is also a representation $(z-y_{st}(a^m-b^n))_{D_a}$ of length
  at most $m$.

  Let $x$ be any integer such that some representation $(x)_{D_b}$ goes to state $s$ in
  $M_b$.  Recall that $(x_{st})_{D_b}$ and $(y_{st})_{D_b}$ go to state $s$ in $M_b$
  and $(x_{st})_{D_a}$ and $(y_{st})_{D_a}$ go to state $t$ in $M_a$.  If $f$ and $g$
  agree on a sufficiently long prefix (to be specified later), then we have
  \begin{align*}
    f_{xb^n+z} &= f_{y_{st}b^n+z} &\text{(since $(x)_{D_b}$ and $(y_{st})_{D_b}$ go to
    state $s$ in $M_b$)}\\
    &= f_{y_{st}a^m + z - y_{st}(a^m-b^n)} &\text{(rewriting the index)}\\
    &= g_{y_{st}a^m + z - y_{st}(a^m-b^n)} &\text{(since $f$ and $g$ agree)}\\
    &= g_{x_{st}a^m + z - y_{st}(a^m-b^n)} &\text{(since $(y_{st})_{D_a}$ and
                                             $(x_{st})_{D_a}$ go to state $t$ in $M_a$)}\\
    &= g_{x_{st}b^n + z + p_{st}} &\text{(rewriting the index)}\\
    &= f_{x_{st}b^n + z + p_{st}} &\text{(since $f$ and $g$ agree)}\\
    &= f_{xb^n + z + p_{st}} &\text{(since $(x_{st})_{D_b}$ and $(x)_{D_b}$ go to
                               state $s$ in $M_b$)}.
  \end{align*}
  For this calculation to be correct, the two sequences $f$
  and $g$ should agree on a prefix of length
  \begin{align*}
    \max\{y_{st},x_{st} \,|\, s\in S_\infty\}a^m + z - y_{st}(a^m-b^n) &\leq
    (\xi-1)a^m + z - y_{st}(a^m-b^n)\\
    &\leq (\xi-1)a^m + 2a^m\\
    &\leq (\xi+1)a^m.
  \end{align*}
  Now $\xi \leq 2b^{|S_b|(|S_a|+1)+1}+1$, so we have
  \begin{align*}
    (\xi+1)a^m &\leq \left(2b^{|S_b|(|S_a|+1)+1}+2\right)a^m \\
    &\leq \left(2b^{|S_b|(|S_a|+1)+1}+2\right)a^{E(a,b,|S_a|,|S_b|)}\\
    &\leq A(a,b,|S_a|,|S_b|).
  \end{align*}

  Thus, if $f$ and $g$ agree on a prefix of length
  $A(a,b,|S_a|,|S_b|)$, then $f$ has a local period
  \[
  p_{st} \leq \frac16 b^n \leq \frac16 b^{E(a,b,|S_a|,|S_b|)}
  \] on the
  interval $[(x+1/3)b^n,(x+5/3)b^n]$.  By the same argument as in
  \cite{Kre18}, the sequence $f$ is ultimately periodic.  We will show
  further that the periodicity begins after a prefix of length at most
  \[
  \left(2b^{|S_b|+1}+\frac13\right)b^n \leq
  \left(2b^{|S_b|+1}+\frac13\right)b^{E(a,b,|S_a|,|S_b|)}.
  \]
  Any representation $(x)_{D_b}$ of length $|S_b|$ must reach a state in $S_\infty$.
  Therefore if $x = 2b^{|S_b|+1}$, then for every $y \geq x$, every
  representation $(y)_{D_b}$ reaches a state in $S_\infty$.  Now by the argument
  of \cite{Kre18}, the sequence $f$ has period $p_f := p_{st}$ starting from
  index $$i_f := \left(2b^{|S_b|+1}+\frac13\right)b^n.$$

  By a similar argument (with the roles of $f$ and $g$ reversed) we
  find that if $f$ and $g$ agree on a prefix of length
  $A(a,b,|S_a|,|S_b|)$, then $g$ has period $p_g$ starting from some
  index $i_g$, where $p_g$ and $i_g$ are defined analogously to $p_f$
  and $i_f$.  Now, we have
  \[
  \max\{p_f, p_g\} \leq \frac16 \theta^{E(a,b,|S_a|,|S_b|)},
  \]
and
  \begin{align*}
  \max\{i_f,i_g\} &\leq \max
  \left\{\left(2b^{|S_b|+1}+\frac13\right)b^{E(a,b,|S_a|,|S_b|)},
  \left(2a^{|S_a|+1}+\frac13\right)a^{E(b,a,|S_b|,|S_a|)} \right\}\\
  &\leq \max \left\{\left(2b^{|S_b|+1}+\frac13\right),
    \left(2a^{|S_a|+1}+\frac13\right) \right\}
    \theta^{E(a,b,|S_a|,|S_b|)},
  \end{align*}
  so
  \begin{align*}
  \max\{i_f, i_g\} + p_f + p_g &\leq \max \left\{\left(2b^{|S_b|+1}+\frac23\right),
    \left(2a^{|S_a|+1}+\frac23\right) \right\} \theta^{E(a,b,|S_a|,|S_b|)}\\
  &\leq A(a,b,|S_a|,|S_b|).
  \end{align*}
  Therefore, by the Fine--Wilf theorem~\cite[Theorem~2.3.5]{Sha09}, the sequences $f$ and $g$ are
  equal.
\end{proof}

In the next corollary, let $\exp_r(x)$ denote the function $r^x$.

\begin{corollary}
  Let $a,b \geq 2$ be multiplicatively independent integers. 
  Let $g = (g_x)_{x \in \mathbb{N}}$ and
  $f = (f_x)_{x \in \mathbb{N}}$ be sequences over a set $\Delta$ of
  size $d$.  Suppose that $g$ is computed by an $a$-DFAO $M_a'$ with
  $R$ states and $f$ is computed by a $b$-DFAO $M_b'$ with $S$ states.
  There is a positive constant $C$ (depending only on $a$ and $b$)
  such that if $f$ and $g$ agree on a prefix of length
  \begin{equation}\label{nice_bound}
  \exp_\theta(\exp_\theta(CR^4S^4))
  \end{equation}
  then $f$ and $g$ are equal and ultimately periodic.
\end{corollary}

\begin{proof}
We have previously observed that conversion from a $b$-DFAO to a
$(b, D_b)$-DFAO increases the number of states to at most the
quantity \eqref{conversion_blowup}.  We apply the bound of
Theorem~\ref{thm:common_prefix} with
\[
|S_a| = R^4\text{ and } |S_b| = S^4.
\]
Simplifying the resulting expression, we find that
there is a positive constant $C$ such that the bound of
Theorem~\ref{thm:common_prefix} is at most the quantity
\eqref{nice_bound}.
\end{proof}

Note that the bound on the length of the common prefix that we obtain
seems absurdly large compared to what seems likely to be the optimal
bound.  It is not too difficult to give an example where the common
prefix has length that is (singly) exponential in the size of the
defining automata.  For instance, let $g$ be the constant (and hence
$a$-automatic) sequence $(0,0,0,0,\ldots)$.  Fix some positive integer
$N$ and let $f$ be the characteristic sequence of the set $\{b^n-1 : n
\geq N\}$.  Then $f$ is an aperiodic $b$-automatic sequence.  Indeed,
a $b$-DFAO $M$ generating $f$ can be obtained from the $N+2$ state DFA
accepting the regular language $0^*(b-1)^N(b-1)^*$ by making the
accepting state output $1$ and all other states output $0$.  Then $M$
has $N+2$ states and the sequences $f$ and $g$ agree on a prefix of
length $b^N-1$.

Now we examine the connection to \emph{automaticity}.
The \emph{$b$-automaticity} of a sequence $g$ is the function
$A_g^b(n)$ whose value is the least number of states required in a
$b$-DFAO that computes a prefix of $g$ of length $n$.
Shallit \cite[Proposition 1(c)]{Sha96} showed that if $g$ is not
$b$-automatic, then there is a constant $c$ such that $A_g^b(n) \geq
c\log_b n$ for infinitely many $n$.

\begin{corollary}
 Let $a,b$ be multiplicatively independent integers with $a,b \geq 2$.
 There is a positive constant $D$ such that the $b$-automaticity
 $A_g^b(n)$ of an  aperiodic $a$-automatic sequence $g$
 satisfies $$A_g^b(n) >  D(\log \log n)^{1/4},$$
 for all $n \geq 0$.
\end{corollary}

\begin{proof}
  Let $M_a$ be an $a$-DFAO computing $g$ and let $M_{b,n}$ be a
  $b$-DFAO computing a sequence $f$ that agrees with
  $g$ on a prefix of length $n$.  Suppose that $M_a$ has $E$ states
  and that $M_{b,n}$ has $S_n$ states.  
  Since $g$ is aperiodic, by \eqref{nice_bound} we have
  \[
  n < \exp_\theta(\exp_\theta(CE^4(S_n)^4))
  \]
  Treating $E$ as a constant, we get
  \[
  S_n > \left(\frac{1}{C^{1/4}E}\right) (\log_\theta \log_\theta n)^{1/4}\\
  = D(\log \log n)^{1/4},
  \]
  for some positive constant $D$.
\end{proof}

Note that while this may seem weaker than the $c\log_b n$ lower bound
mentioned previously, the former only holds for infinitely many $n$,
whereas our lower bound holds for all $n$.  Without the assumption
that $g$ is $a$-automatic, the $b$-automaticity of $g$ could
potentially be constant for long stretches, and only for very sparsely
distributed values of $n$ satisfy $A_g^b(n) \geq c\log_b n$.  Our
result shows that under the assumption that $g$ is $a$-automatic,
the function $A_g^b(n)$ cannot be constant for too long.

On the other hand, our lower bound on the $b$-automaticity does seem
to be rather weak compared to what can be proved for specific sequences.
Shallit \cite{Sha96} showed that if $p$ is the fixed point of $0 \to 01$,
$1 \to 00$, then for $k$ odd, we have $A_p^k(n) = \Omega(n^{1/2}/k)$,
and if $t$ is the fixed point of $0 \to 01$, $1 \to 10$ (the
Thue--Morse word), then for $k$ odd, we have $A_t^k(n) =
\Omega(n^{1/4}/k^{1/2})$.

\section{Common factors of $b$-automatic and Sturmian\\ sequences}\label{sec:sturmian}
As mentioned in the introduction, the problem of bounding the length
of the longest common prefix of a $b$-automatic sequence and a
Sturmian sequence was addressed by Shallit \cite{Sha96}.  In this
section, we show that two such sequences cannot have arbitrarily large
factors in common.  Our main result is the following:

\begin{theorem}\label{main}
Let $f$ be a $b$-automatic sequence and let $g$ be a Sturmian
sequence.  There exists a constant $C$ (depending on $f$ and $g$) such
that if $f$ and $g$ have a factor in common of length $n$, then $n \leq
C$.
\end{theorem}

Note that this result would follow fairly easily from the frequency
results mentioned previously, \emph{if} $f$ is \emph{uniformly
  recurrent} (meaning that every factor $z$ of $f$ occurs infinitely
often, and with bounded gap size between two consecutive occurrences).
However, unlike Sturmian sequences, automatic sequences need not be
uniformly recurrent: consider, for example, the $2$-automatic sequence
that is the characteristic sequence of the powers of $2$.  Our proof
is therefore based on the finiteness of the \emph{$b$-kernel} of
$f$, along with the \emph{uniform distribution} property of
Sturmian sequences (this is similar to the techniques used in 
\cite{Sha96}).

\begin{proof}
Let $f=f_0f_1 \cdots$ and $g=g_0g_1\cdots$, where $g$ has slope
$\alpha$ and intercept $\beta$.  Since the factors of a
Sturmian word do not depend on $\beta$, without loss of
generality, we may suppose that $\beta = 0$ (or, in other words, that
$g$ is a \emph{characteristic word}).  Then $g$ can be defined by the
following rule:
\[
g_n = \begin{cases}1, &\text{ if } \{(n+1)\alpha\} < \alpha;\\
0, &\text{ otherwise.}\end{cases}
\]
Here $\{\cdot\}$ denotes the fractional part of a real number.

Suppose that for some integer $L$, the words $f$ and $g$ have a factor of
length $L$ in common: i.e., for some $i \leq j$, we have
\[ f_i \cdots f_{i+L-1} = g_j \cdots g_{j+L-1}. \]
(We may assume that $i \leq j$ since $g$ is recurrent, but this is not
important for what follows.)  Suppose that the $b$-kernel of $f$,
\[
\{ (f_{nb^r+s})_{n \geq 0} : r \geq 0 \text{ and } 0 \leq s < b^r \},
\]
has $Q$ distinct elements.  Let $r$ satisfy $b^r > Q$.  There there
exist integers $s_1, s_2$ with $0 \leq s_1 < s_2 < b^r$ such that
\[
(f_{nb^r+s_1})_{n \geq 0} = (f_{nb^r+s_2})_{n \geq 0}.
\]

Define
\begin{align*}
d_1 &:= s_1 + j-i + 1,\\
d_2 &:= s_2 + j-i + 1.
\end{align*}
For all $n$ satisfying $i \leq nb^r+s_1$ and $nb^r+s_2 \leq i+L-1$ we
have $f_{nb^r+s_1} = g_{nb^r+d_1-1}$ and $f_{nb^r+s_2} =
g_{nb^r+d_2-1}$.  Since $f_{nb^r+s_1} = f_{nb^r+s_2}$, we have $g_{nb^r+d_1-1} =
g_{nb^r+d_2-1}$.  This means that either the inequalities
\begin{equation}\label{both0}
\{(nb^r+d_1)\alpha\} < \alpha \text{ and } \{(nb^r+d_2)\alpha\} < \alpha
\end{equation}
both hold, or the inequalities
\begin{equation}\label{both1}
\{(nb^r+d_1)\alpha\} \geq \alpha \text{ and } \{(nb^r+d_2)\alpha\} \geq \alpha
\end{equation}
both hold.

If $L$ is arbitrarily large, then there exist arbitrarily large sets
$I$ of consecutive positive integers such that every $n\in I$
satisfies either \eqref{both0} or \eqref{both1}.  Without loss of
generality, suppose that $\{d_2\alpha\} > \{d_1\alpha\}$. Choose
$\epsilon>0$ such that $\epsilon < \{d_2\alpha\} - \{d_1\alpha\}$.
Note that $d_2-d_1 = s_2 - s_1$, so $\epsilon$ does not depend on $L$
(or $I$).  Since $b^r\alpha$ is irrational, if $I$ is sufficiently
large, then by Kronecker's theorem (which asserts that the set of
points $\{n\alpha\}$ is dense in $(0,1)$) there exists $N \in I$ such
that
\[
\{N(b^r\alpha)+d_2\alpha\} \in [\alpha,\alpha+\epsilon].
\]
By the choice of $\epsilon$, this implies that
\[
\{N(b^r\alpha)+d_2\alpha\} \geq \alpha \text{ and }
\{N(b^r\alpha)+d_1\alpha\} < \alpha,
\]
contradicting the assumption that $N$ satisfies one of \eqref{both0}
or \eqref{both1}.  The contradiction means that $L$ must be bounded by
some constant $C$, which proves the theorem.
\end{proof}

\begin{example}
Consider the Thue-Morse word $t = 01101001\cdots$,
and the Fibonacci word $f = 01001010\cdots$ given by the
fixed point of $0 \rightarrow 01$ and $1 \rightarrow 0$.  The latter
is Sturmian.  The set of common factors is
\begin{multline*}
\{ \epsilon, 0, 1, 00, 01, 10, 001, 010, 100, 101, 0010, 0100, 0101, 1001, 1010, \\ 
00101, 01001, 10010, 10100,
010010, 100101, 101001, 0100101, 1010010, 10100101 \},
\end{multline*}
so $C = 8$.
\end{example}

\section{Final thoughts}
As noted at the end of Section~\ref{sec:cobham}, the
$\Omega((\log \log n)^{1/4})$ lower
bound we obtain on the $b$-automaticity of an aperiodic $a$-automatic
sequence is surely not optimal.  Sequences with $O(\log n)$ (i.e., ``low'')
$b$-automaticity are called \emph{$b$-quasiautomatic} in
\cite{Sha96}.  It seems unlikely that an aperiodic $a$-automatic
sequence can even be $b$-quasiautomatic.  Known examples of
$b$-quasiautomatic sequences strongly resemble $b$-automatic
sequences.  For example, the fixed point of the morphism $c \to cba$,
$a \to aa$, $b \to b$, starting with $c$, is $2$-quasiautomatic, but
not $2$-automatic \cite{Sha96}.  Similarly, the fixed point of $1 \to
121$, $2 \to 12221$ is not $2$-automatic \cite{AAS06} but is
conjectured to be $2$-quasiautomatic \cite{Sha96}.  Curiously, this
latter sequence is automatic with respect to the positional numeration
system (and a certain choice of canonical representations) whose place
values are given by the sequence $((2^n - (-1)^n)/3)_{n \geq 1}$
\cite{AAS06}.

We conclude by again mentioning the problem stated in the
introduction of characterizing the common factors of a $b$-automatic
sequence and an $a$-automatic sequence.  Can the method of Krebs be
applied to this problem?

\section*{Acknowledgments}

We thank Jean-Paul Allouche for helpful discussions.  The normalization construction of Section~\ref{sec:normalization} was obtained in discussions with
\'Emilie Charlier, Julien Leroy, and Michel Rigo of the University of
Li\`ege.  We thank them for their help with this problem.

After we posted an initial version of this paper on the arXiv, Thijmen Krebs
contacted us with a number of very helpful comments.  He
clarified some important points regarding his paper, and gave several
suggestions which greatly simplified and improved the presentation of
the normalization construction.  We are very grateful for his
feedback, which significantly improved the exposition in
Section~\ref{sec:normalization}.

\end{document}